\documentclass{acm_proc_article-sp}

\usepackage{url}
\makeatletter
\def\url@leostyle{%
  \@ifundefined{selectfont}{\def\UrlFont{\sf}}{\def\UrlFont{\small\ttfamily}}}
\makeatother
\usepackage[ruled]{algorithm2e}
\usepackage{subfigure}
\newtheorem{example}{Example}
\usepackage{algorithmic}
\newtheorem{theorem}{Theorem}
\renewcommand{\algorithmcfname}{ALGORITHM}
\SetAlFnt{\small}
\SetAlCapFnt{\small}
\SetAlCapNameFnt{\small}
\SetAlCapHSkip{0pt}
\IncMargin{-\parindent}

\newif\ifpdf
\ifx\pdfoutput\undefined
   \pdffalse
\else
   \pdfoutput=1
   \pdftrue
\fi
\ifpdf
   \usepackage{graphicx}
   \usepackage{epstopdf}
   \usepackage{verbatim}
\usepackage{floatrow}
\usepackage{subfigure}
\else
   \usepackage{graphicx}
\fi

\begin{document}

\title{An index for regular expression queries:\\
Design and implementation
}
\subtitle{[Experiment report]
\titlenote{A short version of this paper is published in the 
20th ACM Conference on Information and Knowledge Management (CIKM2011) 
at \textit{at http://www.cikm2011.org/ and ACM DL}
}}

%
%
%
%
%

\numberofauthors{2} 
%
\author{
%
%
\alignauthor
Dominic Tsang\\
       \affaddr{University of Sydney}\\
       \affaddr{City Road}\\
       \affaddr{Sydney, Australia}\\
       \email{dtsa3820@uni.sydney.edu.au}
\alignauthor
Sanjay Chawla\\
       \affaddr{University of Sydney}\\
       \affaddr{City Road}\\
       \affaddr{Sydney, Australia}\\
       \email{sanjay.chawla@uni.sydney.edu.au}
}
\additionalauthors{Additional authors: John Smith (The Th{\o}rv{\"a}ld Group,
email: {\texttt{jsmith@affiliation.org}}) and Julius P.~Kumquat
(The Kumquat Consortium, email: {\texttt{jpkumquat@consortium.net}}).}
\date{30 July 1999}

\maketitle

\begin{abstract}
The {\tt like} regular expression predicate has been part of the SQL standard
since at least 1989. However, despite its popularity and wide usage, database
vendors provide only limited  indexing support for regular expression queries which
almost always require a full table scan. 

In this paper we propose a rigorous and robust approach for providing indexing
support for regular expression queries. Our approach consists of formulating
the indexing problem as a combinatorial optimization problem. We begin
with a database,  abstracted as a collection of strings. From this data set we 
generate a query workload.  The input to the optimization problem is the database
and the workload. The output is a set of multigrams (substrings) 
which can be used as keys to records which satisfy the query workload. 
The multigrams can then
be integrated with the data structure (like B+ trees) to provide indexing
support for the queries. We provide a deterministic and a randomized approximation 
algorithm (with provable guarantees)  to solve the optimization problem. 
Extensive experiments on synthetic data sets demonstrate that our approach is accurate and efficient.

We also present a case study on PROSITE patterns - which are complex regular 
expression signatures for classes of proteins. Again, we are able to demonstrate 
the utility of our indexing approach in terms of accuracy and efficiency. Thus, perhaps for the first time, there is
a robust and practical  indexing mechanism for an important class of database queries. \\[2ex]
\vspace{1cm}

\end{abstract}

\section{Introduction}

Consider a simple database query: 

{\tt  SELECT doc.id FROM doc where doc.text LIKE '\%har\%'}

Current database systems have to carry out a full table scan to answer the above query.
For large databases (like collections of documents) this can be extremely time consuming
rendering the use of regular expression queries almost infeasible.

In the above query, the query poser may be searching for documents which
contain the text {\tt share} or {\tt shard} or something else. To speed up query processing, the database
designer could potentially create an index where the keys are multigrams (substrings) 
which point to all records which contain that multigram - the {\tt posting list} of
the multigram. In this instance, there are several choices that could be made. For example,
 the index algorithm may decide to select the  multigram {\tt har} or {\tt ha} or {\tt ar} or nothing at all. 

To process the query, the query engine  will  determine if there is any of the text fragments
in the regular expression query belong to the key set of the index. If so, then the index will be scanned to 
retrieve the posting list for that multigram. The query engine then will apply
the regular expression template on each record of the posting list and select those that
are matched.  The query processing framework is shown in Figure ~\ref{fig:lpms_qry}.

In Figure \ref{fig:lpms_qry}, the diagram shows the simplified process flow view
of our regular expression query framework. It can be seen that the index is
made up of two components: \textbf{(1) multigrams} - index keys; 
\textbf{(2) posting lists} - database record filter list. 
When user submits a query request, the regular expression will be
firstly submitted to the regular expression rule engine in step one. 
At the same time, the query will be broken down into query keys in step two. 
Then, these query keys will be used to match the multigrams 
as well as the posting list in step three, four and five.
Then, it will generate a candidate database
record set to answer the regular expression query. In step six,
the candidate database records are passed to the regular expression rule
engine, which has already incorporated with the regular expression of interest
in step one, and finally produces the query answer in step seven.

If no index is employed to answer the query request, step two to step five
will be ignored, and the entire database will be passed to the regular
expression engine to process in step six. Generally, this is named
as ``full table scan'' in the database community. 

The key challenge is in deciding what to index. For example, the advantage of indexing {\tt har} over {\tt ha} is efficiency: 
the database records which contain the text {\tt har} is a subset of
those which contain {\tt ha} resulting in lower IO cost for queries which contain {\tt har}. The
disadvantage is that fewer queries are likely to use this key compared to {\tt ha}.
Thus the index designer has to (i) control the size of the index by bounding the maximum
number of multigrams that can be indexed, (ii) balance the trade-off between the likelihood
of a multigram being used by a query and the size of its posting list. 

\begin{enumerate}
\item We present a novel algorithm which selects the multigrams to index. Note that the
number of possible multigrams is infinite. Our approach is to cast the multigram selection
problem as an integer program and then use a linear programming relaxation, followed
by rounding to select the multigrams.
\item Using the above algorithm, we have implemented a fully functional regular
expression query framework on top of a commercial database.
\item We have carried out extensive experiments to test for the accuracy, efficiency
 and robustness of the regular expression querying framework.
\item We present a novel case study using protein data where regular expression patterns
are used routinely to classify families of proteins.
\end{enumerate}

The rest of the paper is as follows. In Section 2 we report on related work. Section 3
formalizes the index selection problem in an integer programming framework and presents 
a small example to illustrate how the algorithm selects the multigrams that are
indexed. Section 4 presents the Linear Programming Multigram Selection (LPMS) algorithm
and the related theory. In Section 5 we report on the extensive set of experiments to
test the query framework. We also present a case study in a real application. We conclude
in Section 6 with a summary and directions for future research. Appendix 1 describes all
the data sets that were created and used for the experiments.

\begin{figure}[t]
    \centering
    \includegraphics[width=0.7\textwidth]{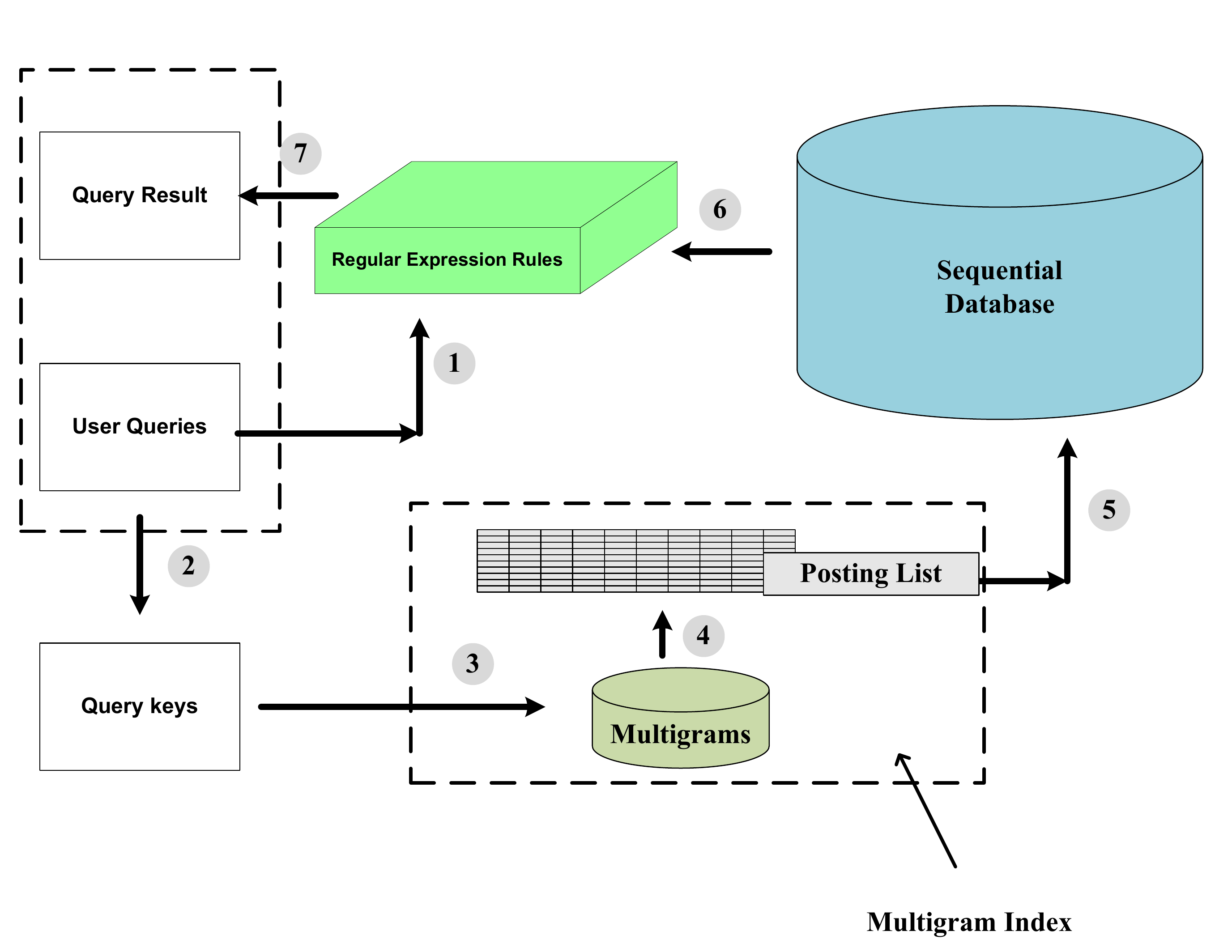}
    \caption{A generic indexed query framework consists of $7$ steps. Whereas, if no indexes
are employed, steps $2$ to $5$ will be skipped}
    \label{fig:lpms_qry}
\end{figure}

\section{Related Work}
Traditionally a regular expression query is processed by constructing 
a non-deterministic finite automata (NFA) which recognizes the language 
defined by the regular expression pattern  \cite{RE-Idx-automaton3,RE-Idx-automaton1,RE-Idx-automaton2}. 
The entire database is processed one character at a time resulting in
$O(mn)$ time and $O(m^2)$ space, where $m$ and $n$ are the size of the regular expression
and database respectively.

The two most commonly indexes which have been currently using in both database systems are B+trees \cite{Idx-Btree2,Idx-Btree1}
and bitmaps \cite{Idx-Btree3}. Both of them are designed to support exact pattern matching rather than
regular expression querying. Therefore, some database vendor tailor these indexes to support regular expression
query, such as: full text indexing in PostgreSQL, and function-based index in Oracle. 
However, their application is criticized as quite restricted.\\

From an indexing perspective,  tries and suffix trees 
\cite{RE-Idx-trie0, RE-Idx-suffixtree1, RE-Idx-suffixtree2, RE-Idx-suffixtree3, RE-Idx-suffixtree4, RE-Idx-suffixtree5}, which are lexicographical  ordered
trees, can be used to support regular expression queries. However, the major weakness
of tries and related structures like suffix-trees is that the size of the index is often
much larger than the size of the database that is being indexed. 

Basically, its structure is categorized as a lexicographical ordered tree. 
The height of a trie, $H(n)$, represents the longest path from the root to a 
leaf node where $n$ is the database size. For a random and uniform 
distributed string database, $H(n)=$ $2log_2(n) + o(log_2(n))$ \cite{RE-Idx-trie0-height}. 
Therefore, we can see that trie is feasible to achieve sub-linear query run time. 
In fact, it is reported that in the worse case, 
the number of nodes (i.e. query space) could be $O(n^2)$ \cite{RE-Idx-trie1}. 
In other words, it indicates that \textbf{trie} improves the query time at the expense of the \textbf{query space}.
This may result in creating the thorny issues such as \textbf{memory trashing} and \textbf{bottleneck}
especially in conditions where the database size is considerably large.\\

In order to reduce the number of node while sustaining trie's performance, 
Weiner \cite{RE-Idx-suffixtree3} transformed Trie to the first suffix tree structure,
that the index is built by processing the database string from right to left which takes $O(nH(n))$ time. 
Later, Ukkonen \cite{RE-Idx-suffixtree1} developed a left-to-right built algorithm
that maintains a suffix tree which takes $O(n)$ time. 
In later research, this suffix tree was later augmented with suffix link \cite{RE-Idx-suffixtree5} to
form a directed acyclic word graph which leads to an algorithm for the construction of the automata.\\

In another stream of research, inverted files \cite{Idx-invertfile} and 
q-grams \cite{Idx-q-grams1,Idx-q-grams2,Idx-q-grams3} are two designs which
support information retrieval and informatics. 
However, both are not suitable for regular expression queries because they rely 
on a predefined list of words.

More recently, a relatively different research strand  has been developed to support
regular expression querying.  The initial work in this area was carried out by Cho and Rajgopalan - the authors who denote their index as FREE multigram index \cite{mgram-free}.


The principle of the FREE model is direct and straightforward. 
The objective is to select the minimal useful multigram set. 
A  multigrams is  considered `\textbf{useful}'  as long as the number of data records containing the multigrams that are smaller than a threshold
(namely selectivity or support).  Otherwise it is considered `\textbf{useless}'.
The FREE algorithm consists of a  series of iterations starting with a single character multigram ($k=1$),  whereas in each iteration, all minimal useful multigrams of length $k$ are then selected.  The remaining are considered `useless'  and form the 
the prefix-seeds of the multigrams in the next iteration.
The process is continually repeated until there is no `useless' multigram is left. 
The algorithm is designed so that the set of useful multigrams selected is  prefix-free.
 
The key insight of the FREE algorithm is that {\it the size of the prefix-free multigram index is always bounded above by the database size.}

This is an important property that all regular expression indexes should
 enforce. Since  FREE selects low support multigrams, queries which use
FREE always take less time than queries which employ a full table scan.
However the  disadvantage of  FREE is that it  does not take any query workload into account and
thus many queries are unlikely to utilize the index. 

To overcome the weakness of FREE,  Hore et. al. \cite{mgram-best}  (in a 2004 CIKM paper) proposed
a multigram selection algorithm called BEST which takes both the database a query workload
into consideration.  In the BEST  algorithm, each multigram is associated with a cost factor $c$ (equivalent to the support) and at the same time associated with the query set by a benefit factor, $b$.  The benefit of the multigram is equal to the number of records that can be pruned
when  utilized by a query. The ratio of $\frac{b}{c}$  forms the \textbf{utility value} which is defined as the objective function to optimize both the index efficiency and query hit rate at the same time.

Hore formalizes the multigram selection problem as an instance of the 
Budgeted Maximum Coverage (BMC) \cite{bmc-1} problem. 
Specifically the cover set forms the `Budgeted' part while the utility
of the index is captured in the `Coverage' component. The main
principle of the BEST algorithm is to select multigrams which  increase
the index hit rate so more queries can utilize the algorithm. However, a major
weakness of the BEST algorithm is that it neither scales to large datasets nor
large query workloads because it uses a cartesian product of query workload and
database as the search space.

Our approach is inspired by both FREE and BEST. Like the BEST algorithm we take both
the database and query workload into consideration but unlike BEST we generate a representative
workload from the database. The advantage of having both a database and workload is that we can formalize the problem in a combinatorial optimization framework.  We also formalize
the problem in a way which ensures that the set of multigrams selected are prefix-free.
Thus we are able to bound the size of the index and guarantee that the running time
of a query using the index will be less than the full database scan.

\section{The index selection problem}
In this section we will use the integer programming framework to formalize the index selection problem. 
We will also provide a small example which explains how the integer programming solution
returns multigrams which can then be indexed. In this section we will assume that the query workload
is given. In practice we generate a representative workload from the database. In Section 4 we
will show how the integer program can be relaxed to return the solutions in an efficient manner.

\subsection{Problem Definition}

To recall we are given a database, abstracted as a collection of strings, and  
a set of regular expression queries $Q$.  Our objective is to select a set of multigrams which will be used as keys of an
index for efficiently answering the queries in $Q$.

The index will be used to retrieve a set of candidate records from the database which
may satisfy the query. After the candidate records are selected, the actual
regular expression matching is done in memory to select the exact set of
records which satisfy the query.

For the purpose of indexing we will treat each query $q$ as
a set of multigrams $M_{q}$. Let $\bar{M}_{q}$ be the set $M_{q}$ and all
the substrings of length at least one  which appear within each multigram in $M_{q}$. 

Let $G = \bigcup_{q \in Q}\bar{M}_{q}$ be the universe of candidate
multigrams. It is from $G$ that we will select a subset $G_{I}$ which will form the index. \\

\begin{example}
Suppose the set $Q$ consists of a single query $q = (ex).\{1,3\}(ess)$. 
Then $M_{q} = \{ex, ess\}$ and  $\bar{M}_{q} = \{ex, es, ss, ess\}$
\end{example}

To formalize the index selection problem as an integer program (IP) we need 
to define (i) integer variables ($x$), (ii) the constraints ($A,b$) and (iii)
the objective function ($c$) and set up the problem as
\begin{equation}
\begin{array}{lll}
\mbox{minimize} & \sum\limits_{g \in G} c_{g}x_{g} & \\[2ex]
\mbox{subject to} & Ax \geq b & \\[2ex]
& x \in \{0,1\}
\end{array}
\end{equation}
For each multigram $g$ in $G$ we associate an integer binary variable $x_{g} \in \{0,1\}$. 
The variable $x_{g}$  will be set to one if the multigram $g$ is selected to be part of the index. 
To formalize the constraints we will construct a $|Q| \times |G|$ matrix A such that

\begin{equation}
A_{i,j} = 
\left\{
\begin{array}{ll}
s(g_j) &  \mbox{if } g_{j} \in \bar{M}_{q_{i}};\\
0      & \mbox{otherwise}.
\end{array}
\label{model_matrix_A}
\right\}
\end{equation}

Here $s(g)$ is the {\it support} or the number of rows in the database
in which the multigram $g$ appears at least once. It is through the
definition of $A$ that the integer program captures the characteristics
of the underlying database. We also need to define the $|Q|$ dimensional
vector $b$ which defines the right hand side of the system $Ax \geq b$.
Here we define

\begin{equation}
\begin{array}{l}
b_{i} = \min\limits_{g \in \bar{M}_{q_{i}}}s(g) \;\; \forall i
\end{array}
\label{model_matrix_b}
\end{equation}

The definition of $b_{i}$ can be interpreted as follows. Each row of
$A$ represents a query. The row constraint captures the smallest number
of database records that must be returned if the query will use the index. This
is captured with multigram $g$ with the smallest
support contained in the {\it expanded} query $\bar{M}_{q}$. 

We now define the objective function which is typically of the
form 

\begin{equation}
\begin{array}{l}
\sum\limits_{g \in G}c_{g}x_{g}
\end{array}
\label{model_constraint}
\end{equation}

The objective should capture the trade-off between the coverage of a multigram
and it is support in a database. By coverage we mean that one multigram
can be used by several queries. On the other hand if the support of this
multigram is high (relative to the size of the database) then the index
is not necessarily useful as a full table scan may only be slightly less
efficient. For example suppose the objective is to select multigrams
to index text documents. Then selecting the word ``the'' as an index
is not very efficient as most documents in the database will be returned
as candidates. In the language of information retrieval the choice of
multigrams need to balance the trade-off between precision versus recall.

Since we want to cast the IP as a minimization problem we want to select
multigrams which have low support and at the same time we want the selected
multigram to be used by as many queries as possible. Thus we define

\begin{equation}
\begin{array}{l}
c_{g} = \frac{s(g)}{|g|\sum\limits_{q \in Q}\imath(g \in \bar{M}_{q})}
\end{array}
\label{model_matrix_c}
\end{equation}

where $\imath(g \in \bar{M}_{q})$ is an indicator function. Thus the numerator
of $c_g$ is the support of the multigram $g$ and we want multigrams with
small support (making the index more efficient) and at the same time we
want to select multigrams with high coverage so  fewer multigrams need
to selected to form the index.

\begin{theorem}The solution vector $x$ returned by the integer program
represents a set of multigrams which form a prefix-free set.
\end{theorem}
\begin{proof}
The proof will be by contradiction.
Suppose $x$ does not return a prefix-free set. Let $x_{u}$ and $x_{v}$
represent multigrams $u$ and $v$ such that $u$ is a prefix of $v$ and
therefore
\begin{enumerate}
\item $s(u) \geq s(v)$.
\item If $v \in \bar{M}_{q}$ then $u \in \bar{M}_{q}$ and
\item $\frac{1}{|u|} \geq \frac{1}{|v|}$
\end{enumerate}
Thus $c_{u} \geq c_{v}$. Now if construct a new vector $x'$ such
\[
x'(g) = \left\{\begin{array}{ll}
x(g) & \mbox{ if} g \neq u \\
0 & \mbox{ otherwise }
\end{array}
\right.
\]
Then $\sum_{g \in G}c_{g}x'_{g} <  \sum_{g \in G}c_{g}x_{g}$. Furthermore by
construction of the constraint matrix $Ax' \geq b$. This violates the
minimality of $x$. Thus the solution returned by the integer program
must be prefix-free.
\end{proof}

The importance of having a prefix-free index set has been noted (and proved)  before
by Cho and Rajgoplan \cite{mgram-free}. Essentially the prefix-freeness guarantees that
the size of the index (as measured by the number of pointers into the
database) is bounded by the size of the database. More formally (and
in our notation)
\begin{theorem} If $x$ is the vector returned by the integer program
and $|R|$ is the size of the underlying database then
\[\sum_{\{g:x_{g}=1\}}s(g) \leq |R|
\]
\end{theorem}
The importance of Theorem 1 and 2 cannot be underestimated. For the first
 time we have a principled way of selecting a set of multigrams to index.
The definition of the problem gives us prefix-freeness. As a result
the constraint on the size of the index is now {\it endogenous} to
the problem as opposed to providing an exogenous (budget) constraint as in
 the BEST method \cite{mgram-best}.

\subsection{An example of index selection with integer programming}
In this  section we will walk through a simple example
to illustrate how the integer programming approach can
be used for multigram selection.

Table \ref{word_db} shows a simplified word database. Table \ref{qry_examples} shows two typical regular expression
queries. For instance, the first query  looks for the words in the word database that
contains the substrings, `ex' or `pr', followed by any substring of length
between 1 and  3 and then immediately followed by the substrings, `eed' or `ess'.
Since the regular expression is composed of an `or' predicate, all the \textbf{or-}ing 
elements within the regular expression must be uniquely identified to avoid ambiguity.
In column two of Table \ref{qry_examples}, we show that the first regular expression query
is broken down into four non-\textbf{or} regular expression queries.

The word database consists of 62 multigrams with length ranging from two to four\footnote{We have restricted multigrams to be between length two and four for illustration. In practice there is no restriction.}
Only 14 multigrams are relevant to the two regular expression queries.
In table \ref{useful_gram}, we lists out all the multigrams
in the word database that relevant for $Q$. 
From equation \ref{model_matrix_A}, \ref{model_matrix_b} \& \ref{model_matrix_c},
we can calculate matrix $A$ and vector $c$ and  $b$.

\begin{table}[t]
\caption{\scriptsize{The word database $R$ contains 8 records made up of the multigram set $\bar{M}_{2,q}$ 
as elaborated in table \ref{useful_gram} }
\label{word_db}}{
\begin{tabular}{|r|r|r|r|r|r|} \hline
$r_1$ & `succeed' & $r_4$ & `recede'   & $r_7$ & `succession'  \\ \hline
$r_2$ & `proceed' & $r_5$ & `secession & $r_8$ & `excess'      \\ \hline
$r_3$ & `precede' & $r_6$ & `exceed'   &       &               \\ \hline
\end{tabular}
}
\end{table}

\begin{table}[t]
\caption{\scriptsize{Query set $Q$ contains 2 regular expression queries
with 14 (out of 64) multigrams}
\label{qry_examples}}{
\begin{tabular}{|r|l|} \hline
           Regular Expression&                 $M_{q}$ \\ \hline
(ex)$|$(pr).\{1,3\}(eed)$|$(ess) & $q_1$:(ex).\{1,3\}(ess)  \\ 
                                 & $q_2$:(ex).\{1,3\}(eed)  \\ 
                                 & $q_3$:(pr).\{1,3\}(ess)  \\ 
                                 & $q_4$:(ex).\{1,3\}(eed)  \\ \hline
(pr)$|$(re).\{1,2\}(cede)        & $q_5$:(pr).\{1,2\}(cede) \\ \hline
\end{tabular}
}
\end{table}

\begin{table}[t]
\caption{\scriptsize{The multigram set $\bar{M}_{2,q}$ for $R$ \& $Q$ contains 64 multigrams
with length ranging from 2 to 4}
\label{useful_gram}}{
\begin{tabular}{|r|r||r|r||r|r|} \hline
      & $g_j$,$s(g_j)$ &        & $g_j$,$s(g_j)$&         & $g_j$,$s(g_j)$\\ \hline \hline
$g_1$ &         `ex', 2&   $g_6$&        `ed', 5& $g_{11}$&      `ced',2  \\ \hline
$g_2$ &         `es', 3&   $g_7$&       `eed', 3& $g_{12}$&      `ede',2  \\ \hline
$g_3$ &         `ss', 3&   $g_8$&        `pr', 2& $g_{13}$&     `cede',2  \\ \hline
$g_4$ &        `ess', 3&   $g_9$&        `ce', 8& $g_{14}$&       `re',2  \\ \hline
$g_5$ &         `ee', 3&$g_{10}$&        `de', 2& $g_{15}$&               \\ \hline
\end{tabular}
}
\end{table}

Therefore, $\large A$ =
\[ \left( \begin{array}{cccccccccccccc}
2&3&3&3&0&0&0&0&0&0&0&0&0&0 \\
2&0&0&0&3&5&3&0&0&0&0&0&0&0 \\ 
0&3&3&3&0&0&0&2&0&0&0&0&0&0 \\
0&0&0&0&3&5&3&2&0&0&0&0&0&0 \\
0&0&0&0&0&5&0&2&8&2&2&2&2&0 \\
0&0&0&0&0&5&0&0&8&2&2&2&2&2 \\
\end{array} \right)\]

, $\large c$ = (1 1.5 1.5 1 1.5 2.5 1 1 4 1 0.6 0.6 0.5 1)
and $\large b$ = (2 2 2 2 2)'.

When the problem instance $(A,b,c)$ are fed into an integer program solver, the 
returned solution vector  $\large x$ = (1 0 0 0 0 0 0 1 0 0 0 0 1 0)'
or $g_1$, $g_8$ \& $g_{13}$. As a result, `ex', `pr', `cede' 
are the optimal set of multigram selected.  

\section{A Linear Programming Relaxation}
In the previous section we have presented a model for the multigram
selection problem based on integer programming. We proved that the
the resulting multigrams form a prefix-free set. 

However theoretically the general integer program problem is NP-Hard and
even in practice algorithms for IP programs require exponential time for 
most instances. This is because unlike linear programming problems where 
convexity of the problem can be exploited (local optima are global optima), for integer programming the full lattice of feasible points has to be explored ~\cite{vazirani}.

A common approach to get around the complexity of integer programs is
relax them to linear programs which can then be solved efficiently (in polynomial time). The solution of the linear program provides
a natural lower bound (for minimization problems) for the integer program.
The linear program {\it fractional} solution then needs to be converted to an integer solution. There are two approaches for carrying out the conversion. The first approach is to select a deterministic threshold and use that
to round the solutions. The second approach is to interpret the solution
vector (typically between 0 and 1) as probabilities and carry out a
randomized rounding.

For the rest of the paper, the solution of multigram selection problem
obtained using the integer program will be referred as IPMS (Integer Program Multigram Selection). Similarly those obtained from deterministic and randomized rounding of the linear program will will form the basis of the  LPMS-D and LPMS-R algorithms respectively.

\subsection{The LP Relaxation}
The LP relaxation of the  integer program is 

\begin{equation}
\begin{array}{lll}
\mbox{minimize} & \sum\limits_{g \in G} c_{g}x_{g} & \\[2ex]
\mbox{subject to} & Ax \geq b & \\[2ex]
&  0 \leq x \leq 1
\end{array}
\end{equation}

Clearly if the optimal solution of the integer program is denoted as $OPT_{I}$
and the solution of the linear program is denoted as $OPT_{LP}$ then $OPT_{LP} \leq OPT_{I}$. 

Now the solution vector returned by the linear program is fractional and so we need to convert the fractional solution into an integer solution. 

\begin{theorem} Let $s_{\max}$ and $s_{\min}$ be the multigram in $G$ with the largest and smallest non-zero support respectively. Furthermore let $m^{\ast} =\max_{q \in Q}|\bar{M_{q}}|$. If $x_{l}$ is the solution of the relaxed
linear program and we convert it to an integer solution by rounding up
all elements in $x_{l}$ which are greater than $\frac{s_{\min}}{s_{\max}m^{\ast}}$ then the rounded solution is a feasible solution of the integer program
and achieves a constant approximation to the original integer program
\end{theorem}
\begin{proof}
All we have to show is that if all elements of $x_{l}$ are less than
$\frac{s_{\min}}{s_{\max}m^{\ast}}$ then $x_{l}$ is not a feasible
solution of the linear program.

Now this is true because for each row $i$ of $Ax \geq b$ there are only
$m^{\ast}$ non-zero entries in each row sum and therefore
\begin{eqnarray*}
\sum\limits_{j}a_{ij}x_{j} & < s_{\max}m^{\ast}\frac{s_{\min}}{s_{\max}m^{\ast}} = s_{\min} \\
\end{eqnarray*}
This contradicts the feasibility of $x_{l}$ (as $Ax \geq b$). 
Thus by rounding all
terms based on the threshold we not only get a feasible solution but
the value of the cost function goes up by at most $\frac{s_{\max}m^{\ast}}{s_{\min}}$ \\
Furthermore the approximation is obtained by noting that

\begin{equation}
OPT_{LP} \leq OPT_{IP} \leq \frac{s_{\max}m^{\ast}}{s_{\min}}OPT_{LP}
\label{determinstic_rounding}
\end{equation}
\end{proof}

Therefore,\\

\begin{equation}
x_{l}^{D} =  
\left\{
\begin{array}{ll}
1 &  \mbox{if } x_{l} > s_{\max}m^{\ast}\frac{s_{\min}}{s_{\max}m^{\ast}}\\[2ex]
0 &  \mbox{otherwise}
\end{array}
\label{deterministic_rounding_formula}
\right\}
\end{equation}

The above theorem will be used as a basis for designing a deterministic
rounding algorithm LPMS-D.

\subsection{LP Relaxation - Randomized Rounding}
While the deterministic rounding algorithm gives an approximate and
feasible solution, the approximation bound is quite loose. A more
practical solution is to use randomized rounding by interpreting the
values of the linear program solution as probabilities. Thus each
component of $x_{l}$ is treated like a probability. The cost function
can be treated as a random variable and it is easy to show that the
expected value of the cost function is equal to value of the linear
program. 

\subsection{Prefix free enforcement}
While LP relaxation and rounding  allows us to overcome the exponential runtime complexity of integer
programs,  it does not  guarantee that the resulting multigram set is 
a prefix-free. The resulting size of the index may not be bounded by the database size. In order to
overcome this shortcoming, we enforce the prefix-free constraint
by proposing an iterative algorithm- LPMS.  

Algorithm 1 outlines the LPMS multigram prefix-free selection
process. The expandSet starts with an element ``.''.
Then  each element in expandSet is  expanded by appending
each character from the alphabet set, $\Sigma$, which will
form the childrenSet. Then all query keys contained
in the query set, $Q$, are consolidated to generate the
$\bar{M}_{q}$ set. After this any multigrams in the
childrenSet that are not present in the $\bar{M}_{q}$ set
will be removed. Then the refined childrenSet and set $Q$
will form the input to  an LP solver to produce the linear
programming solution, $x$. Finally, $x$ is relaxed by
applying either deterministic or randomized rounding. Those 
multigrams with its associated x value equals to ``1'' will be
selected as the multigram and saved in a set $G$. All
unselected multigrams will then replace the current
expandSet. The process will be repeated 
until the expandSet is empty, i.e., all the multigrams are selected.

\begin{algorithm}[t]
\SetAlgoNoLine
\KwIn{string database (R), query set (Q), alphabet set ($\Sigma$)}
\KwOut{Multigram index (G)}
expandSet = \{.\}\;
\Repeat{(expandSet = \{\})}
{
    For each multigram in the expandSet, append a member in $\Sigma$ and save to the childrenSet\;
    Remove all multigram g $\in$ childrenSet if g $\notin$  $\bar{M}_{k,q}$ $\forall$ $k$\;
    Populate matrix A, b, c of the LP model using formulae (\ref{model_matrix_A}), (\ref{model_matrix_b}) 
         and (\ref{model_matrix_c}) respectively\;
    Apply LPMS-D or LPMS-R to find the vector x\;
    Move all multigrams whose associated value in the vector x $=1$ to G\;
    Those multigrams remaining will become the expandSet\; 
}
\caption{LPMS multigram selection algorithm}
\label{alg:one}
\end{algorithm}

\section{Experiments Roadmap}

We have carried out extensive experiments to test our approach for 
accuracy, efficiency  and robustness. The experiment roadmap is shown in Table ~\ref{tab:expt_roadmap}.  
There are total of five experiments which includes one case
study on a real protein  data set. For each of the experiment we generate a different synthetic
data set to vary conditions appropriate for that particular experiment. The details of
data generation are given in Appendix 1.

\subsection{Experiment 1 - Accuracy and Performance}
In this experiment we measure accuracy and performance of our proposed LPMS-D and LPMS-R and compare it with the FREE approach. We use recall as the measure of accuracy which captures the percentage of queries which use the index. For performance we measure precision which
effectively measure  what percentage of candidate records are actually satisfied by a 
query in the workload. We created several data sets by varying the
standard deviation of the support distribution of the multigrams. The details of data  generation are in Appendix 1.

The results are summarized in Figure ~\ref{fig:expt1_recall_precision} and \ref{fig:expt1_qry_cnt_precision}. In 
Figure \ref{fig:expt1_recall_precision} we measure the precision and recall of the LPMS-R, LPMS-D and FREE as a function of the different data sets. As expected (since we proved it) we get a recall of one for LPMS-D. However note that the recall of LPMS-R is also very high -greater than 0.98 in all cases. Thus even after randomized rounding our approach is still able to satisfy most of the constraints of the integer program. Not surprisingly the recall of the FREE approach is very low. This is because  the FREE method uses a set of prefix-free multigrams of low support as the 
basis of the index. The FREE approach does not use the query workload for designing
the index.

When it comes to precision, the story is the opposite. The precision of FREE is greater than
that of LPMS-R which is greater than LPMS-D. Again, FREE indexes multigrams with
small support so when a multigram is selected by a query it is bound to have selected
only a few candidate records making the precision high. 

A more interesting result is to compare the precision of LPMS-D and LPMS-R which 
is shown in Figure \ref{fig:expt1_qry_cnt_precision}. Here the X-axis represents the query-id's
ranked  by their precision on LPMS-R. Thus the query with the highest precision using
LPMS-R is on the left. For each of the query we also calculated the precision of LPMS-D. 
The results clearly show that for almost all instances, the precision of LPMS-R is
higher than that of LPMS-D. 

The results of Figure \ref{fig:expt1_recall_precision} allow us to conclude that LPMS-R is highly accurate (with recall almost equal to that of LPMS-D) and also very efficient (as precision is
a measure of efficieny). 
\begin{table*}[t]
\caption{Road map of the five experiments performed to evaluate the LPMS multigram index}
\label{tab:expt_roadmap}{
\small
\begin{tabular}{|l|l|l|l|l|l|} \hline
Expt&                   Objective and Description&Data Type&   Data Set Description&    Metric& Results \\ \hline \hline
1   & Compare accuracy and performance of LPMS-R,&Synthetic& Data sets with varying& Precision&Fig \ref{fig:expt1_recall_precision} and \ref{fig:expt1_qry_cnt_precision}\\ 
    &  LPMS-D, RDB and and FREE \cite{mgram-free}&         &support distribution of&       and& \\
    &                                            &         &             multigrams&    Recall& \\ \hline
2   &            Scalability of LPMS-R and LPMS-D&Synthetic&  Data sets with varing&Build time&Figure \ref{fig:expt2_linear_db_consturction_time} and  \ref{fig:expt2_linear_qry_consturction_time}\\
    &                                            &         &data and query workload& and Query&  \\
    &                                            &         &                   size&      time&  \\ \hline
3   &      Quality of LPMS-R relaxation vis-a-vis&Synthetic&   Same as Experiment 2&   Size of&Table \ref{expt3_qry_result}\\
    &                 optimal Integer Program and&         &                       &     index&             \\
    &                      BEST \cite{mgram-best}&         &                       &     (posting list) &             \\ \hline
4   &Robustness of index as query workload change&Synthetic& Data sets with varying& Precision&Table \ref{expt4_2pct_qry_testes} \\
    &                                            &         &          alphabet size&       and&             \\
    &                                            &         &                       &    Recall&             \\ \hline
5  & Case Study on Prosite protein patterns using&     Real&Protein database\cite{pfam_site}& Precision&Figure \ref{fig:full_posix_app} and \\
   &                                       LPMS-D&         &            and Prosite pattern &       and&                                    \\
   &                                             &         & query\cite{prosite-1}          &    Recall&             \\ \hline
\end{tabular}}
\end{table*}

\begin{figure*}[t]
    \centering
    \subfigure[]{
\includegraphics[width=0.40\textwidth]{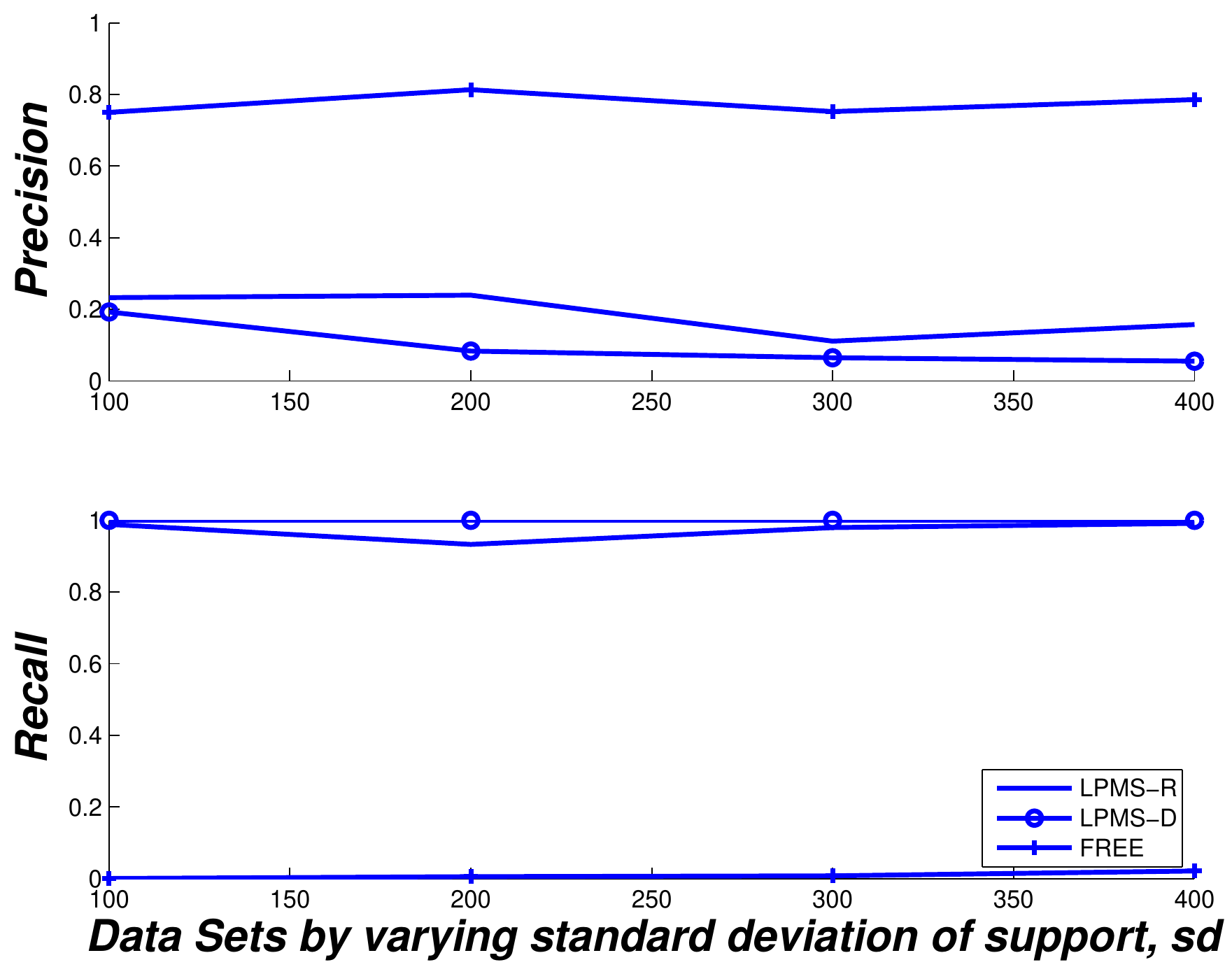}
    \label{fig:expt1_recall_precision}
}
\subfigure[]{
    \centering
    \includegraphics[width=0.40\textwidth]{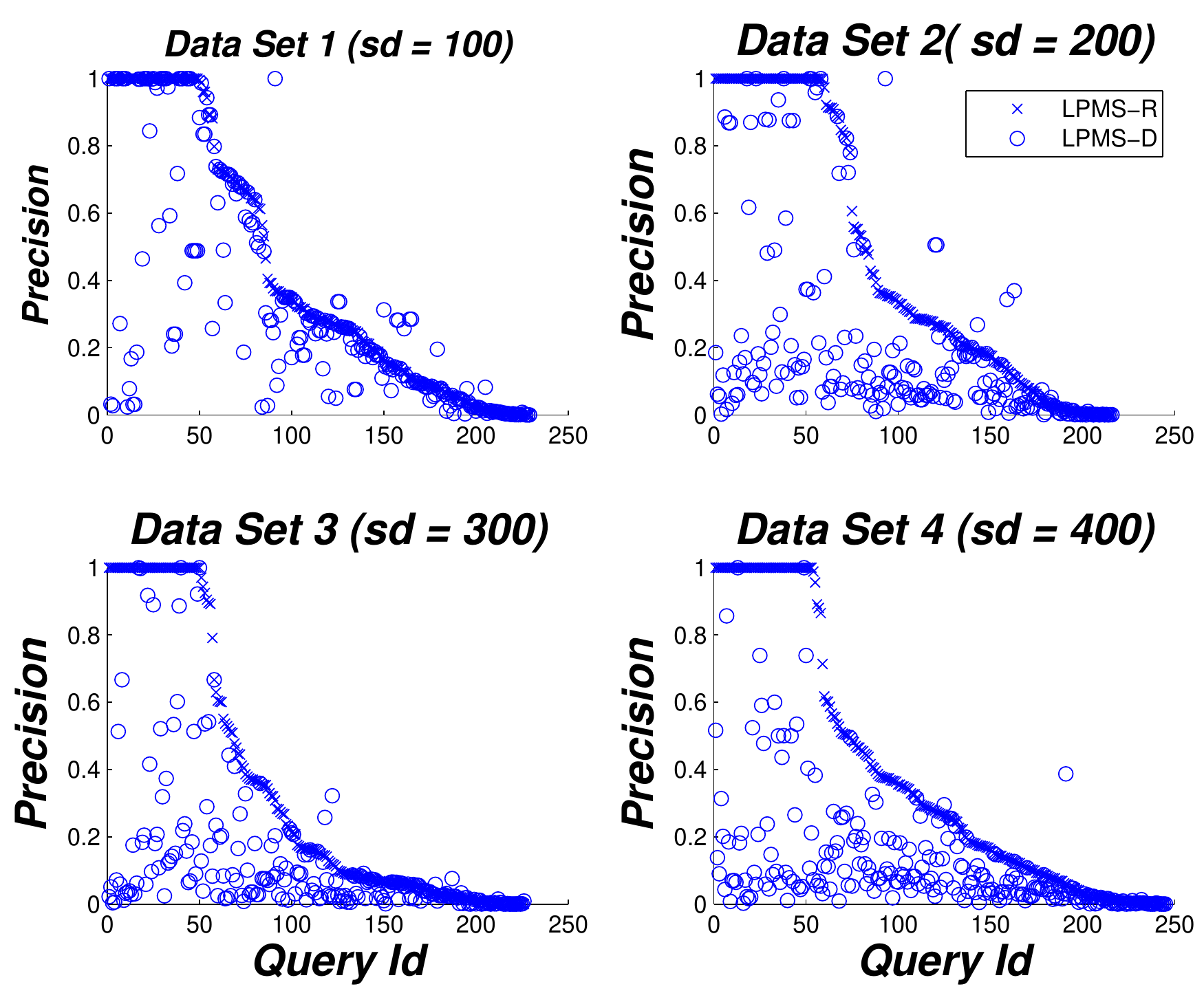}
    \label{fig:expt1_qry_cnt_precision}
}
\caption{Experiment 1: (a) The recall of LPMS-D > LPMS-R $>>$ FREE;  the precision of FREE > LPMS-R > LPMS-D.
(b) For almost all queries the precision of LPMS-R is higher than LPMS-D. This suggests that LPMS-R 
is both highly accurate and yet more efficient than LPMS-D. Note that the recall of FREE is very low making
it ineligible as a serious candidate for indexing.}
\end{figure*}

\subsection{Experiment 2 - Scalability}
In this experiment we measure the  execution time of the componets of Algorithm 1 (LPMS-R).
Algorithm 1 has three distinct components: (i) multigram generation (ii) model construction
and (iii) model solver - which invokes a call to a linear programming engine. 

The multigram generation time (MGT) is the time required to generate all the candidate
multigrams and also calculate their support value. The model construction time (MCT) is 
the time required to populate the matrix $A$, the constraint vector $b$ and 
the cost function $c$.

The execution time is measured  as a  function of varying data size and query workload. In the first set of experiments, the data set size are varied between 20K and 100K while keeping the query workload size
to 1000.   In the second set, the data set size is fixed but the query workload varies between 90K and 500K. Details data set and query workload generation is described in Appendix 1.

FIgure \ref{fig:expt2_linear_db_consturction_time} shows the running time for MGT, the MCT and  overall index construction processing
as a function of increasing database size.  The MGT scales linearly with database size, 
the MCT remains nearly constant and the overall time also scales linearly.   The MCT remains constant as the size of the matrix $A$ and $b$ is dependent on the size of the query workload which is kept constant.

Figure \ref{fig:expt2_linear_qry_consturction_time}, shows the running time as a function of increasing query workload. As expected the MGT is constant (as the database size is kept constant). The MCT scales linearly with query workload size and the overall time increases in a super-linear fashion. This is not surprising because the overall time includes the time to invoke and execute the linear programming
engine. Further refinements, like using a primal-dual approach instead of invoking a linear solver (like simplex) may help reduce the overall time.

However, we can conclude that the current approach scales nearly linearly with both database size and query workload size thus making the proposed approach feasible and practical.

\subsection{Experiment 3 - Optimality}
In Experiment 3 we measure the  divergence from ``optimality'' of LPMS-R and LPMS-D.
The Integer Programming Multigram Selection (IPMS) gives the optimal solution but is only
feasible and practical for small data sets. Thus on small data sets we can experimentally 
compare how far the solution returned from LPMS is from IPMS.  We also compare our approach against the BEST
 algorithm explained earlier \cite{mgram-best}. The BEST algorithm associates a benefit value with each
multigram. We sort the multigrams by benefit value and select the top 100,110,120,130,140,150, 200 and 250
multigrams. The top 250 multigrams result in a hit rate of 100\%, i.e., all the queries in the workload
use the index. Again, the data set which consists of 2000 records and 100 queries is described in
Appendix 1.

The optimality results comparing LPMS-R, IPMS and BEST are shown in Table ~\ref{expt3_qry_result}. 
The IPMS selects multigrams whose posting size is remarkably small - just 106! LPMS-R selects multigrams
with a posting size of 2773. This implies an average precision of 0.304. The posting list of the BEST
algorithm start at 6,774 and by the time the top 250 multigrams are selected, the size of the
posting list has increased to 17,107. Clearly LPMS-R is far superior to the BEST approach.

\begin{table}[h]
\begin{tabular}{|r|r|r|r|r|} \hline 
 Index&\# correct&      Precision& Posting&Prefix\\ 
  type&    query&       Mean/Std&    size&Free  \\ \hline \hline
  IPMS&      100&            1/0&     106&Y     \\ \hline
LPMS-R&       99&    0.304/0.435&   2,773&Y     \\ \hline
 B-100&       91&     0.097/0.28&   6,744&N     \\ \hline
 B-110&       94&     0.067/0.23&   7,396&N     \\ \hline
 B-120&       95&     0.056/0.21&   8,110&N     \\ \hline
 B-130&       97&     0.036/0.17&   8,814&N     \\ \hline
 B-140&       98&     0.025/0.13&   9,513&N     \\ \hline
 B-150&       99&    0.016/0.099&  10,156&N     \\ \hline
 B-200&       99&    0.014/0.099&  13,578&N     \\ \hline
 B-250&      100&    0.004/0.002&  17,107&N     \\ \hline
\end{tabular}
\caption{Experiment 3: On small data sets we can exactly solve for the  integer programming
solution. The optimal posting list size is 106 and LPMS-R returns a list of size of 2,773. The
size of the posting list of BEST, for similar levels of accuracy, is greater than 10,000. Posting
list is a measure of index efficiency.}
\label{expt3_qry_result}
\end{table}

\begin{figure*}[t]
    \centering
    \subfigure[The overall build time for LPMS-R scales linearly with database size]{
    \includegraphics[width=0.4\textwidth]{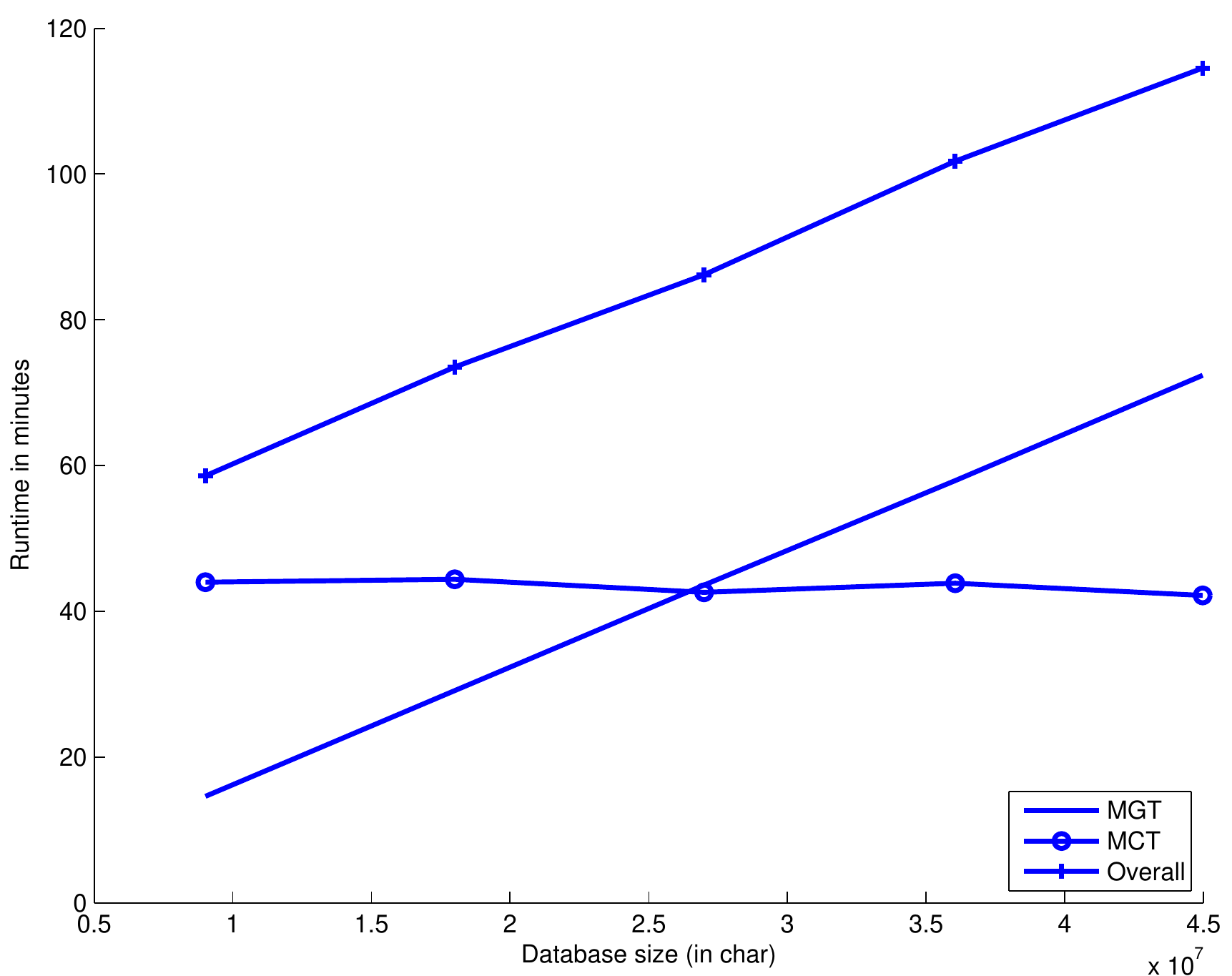}
    \label{fig:expt2_linear_db_consturction_time}
}
\subfigure[The overall build time for LPMS-R scales super-linearly with query workload size]{
    \centering
    \includegraphics[width=0.4\textwidth]{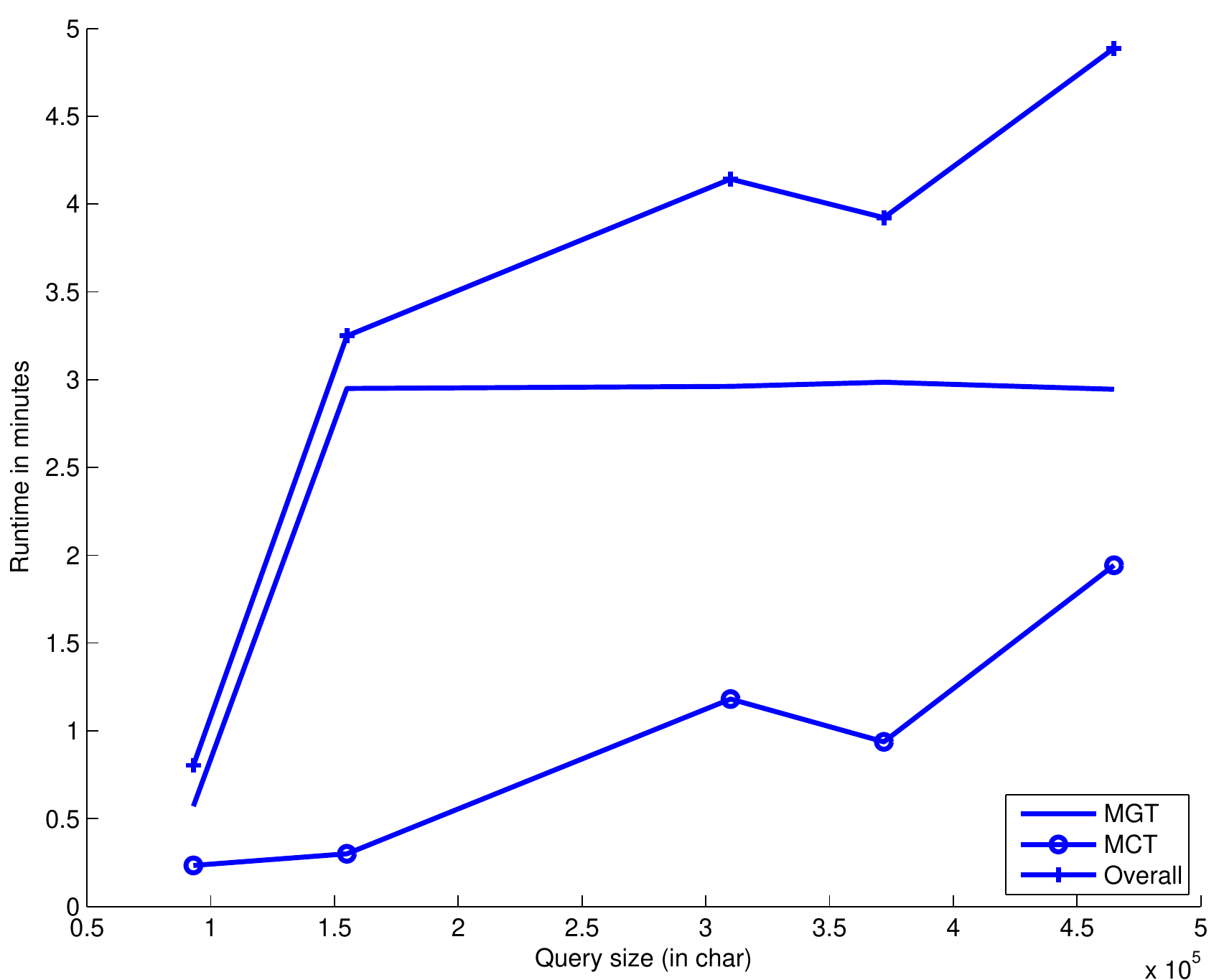}
    \label{fig:expt2_linear_qry_consturction_time}
}
\caption{Experiment 2: Scalability of the LPMS-R construction time as a function of (a) database size; (b) query size. Near-linear scalability suggests that LPMS is a practical method for index generation.}
\end{figure*}

\subsection{Experiment 4: Robustness}
The objective of this experiment is to measure to what degree the index is useful for
queries which were not explictly used for constructing the index. 

We created four different data sets Rob01-Rob04, composed from the alphabet 'A-D', 'A-H', 'A-L'
and 'A-P' respectively. Each data set contains 5,000 records. The details of the data set
are in Appendix 1.

We next defined a query pattern with a key length bounded between three and eight.
For example, the following sample query consists of three query keys and two
gap constraints.

\[
\mbox{
``(HFBFD)(.{0,9})(AEHDCEAG)(.{0,39})(CCCGDAAE)'' }
\]

Using this query pattern, three sample query sets are
generated from the $10\%$, $30\%$ and $50\%$ sample of Rob01,
Rob02, Rob03 and Rob04. Then, LPMS-R and LPMS-D indexes
are constructed using each sample query set and each data set.
This will create twenty-four LPMS multigram indexes.
For example, index `LPMS-R-10' in row two and column 3 of table \ref{expt4_2pct_qry_testes}
is constructed on the data set Rob03, based on the $10\%$ sample query set.\\

For querying we generated five independent test query sets, each from the
$2\%$ sample of Rob01, Rob02, Rob03 and Rob04. This will create
twenty independent test query sets. These test query sets are
tested against the corresponding LPMS-R and LPMS-D indexes.\\

Table \ref{expt4_2pct_qry_testes} reports the test results. It shows
that the recall of all sixty sets of LPMS-D query are equal to $1$
except in a few exceptional occasions. In the LPMS-R
query, the recall decreases as the alphabet size increases. Then,
the recall improves as its corresponding index sample size increases.\\

On the other hand, the precision of the LPMS-R queries are not consistent.
However, the precision of the LPMS-D queries consistently improve as
the sample size increases.\\

These results clearly show that the performance of the LPMS-D multigram
index is depending on the sample size. In other words, for a given 
query pattern, if the sample database size is large enough, the 
LPMS-D multigram index is capable to support arbitrary query and produce
accurate result. In addition, by further increasing the sample size,
the precision will be further improved. This clearly shows that
the LPMS-D is robust.

\begin{table}[t]
\caption{Experiment 4: The recall values of both LPMS-R and LPMS-D 
on test sets on several data sets. See text for explanation. High
recall values suggests that both LPMS-R and LPMS-D are robust against
query perturbations.}
\label{expt4_2pct_qry_testes}
\begin{tabular}{|r|r|r|r|r|}                       \hline 
           &   Rob01 & Rob02 & Rob03 & Rob04 \\ \hline
LPMS-R-10  &       1 &  0.78 &  0.63 &  0.45 \\
LPMS-R-30  &    0.98 &  0.82 &  0.72 &  0.61 \\
LPMS-R-50  &       1 &  0.95 &  0.81 &  0.58 \\ \hline
LPMS-D-10  &       1 &     1 &  0.99 &  0.98 \\
LPMS-D-30  &       1 &     1 &     1 &     1 \\
LPMS-D-50  &       1 &     1 &     1 &     1 \\ \hline
\end{tabular}
\end{table}

\subsection{Case Study: PROSITE Patterns}
So far the previous four experiments have focused on studying the 
performance of the LPMS in a controlled fashion. In this
section, the objective is to apply the LPMS multigram on a real data and contemporary application.

We selected 100K protein sequences from the PFAM-A protein database \cite{pfam_site}. Each sequence 
is composed of 20 distinct amino acids (alphabet) with length ranging from 4 to 2750. The query 
set is made up of $100$ Prosite signatures \cite{prosite-1},
which are downloaded from the PDB \cite{prosite_site}. Each Prosite signature is a regular
expression defining a protein class.

Figure \ref{fig:full_posix_app} shows the data flow of the LPMS-D multigram index Prosite-Protein application.
Step 1 through Step 5 on the right of figure \ref{fig:full_posix_app}   
shows how the LPMS-D multigram index is constructed. While
step A to step E on the left shows how the LPMS-D multigram is
used to process a Prosite query.\\
 
For example, consider the Ribosomal protein S18 (PS000057) from the Prosite pattern \cite{prosite-1}.
\[
\mbox{
``[IVRLP]-[DYN]-[YLF]-x(2,3)-...[RHG]-[LIVMASR]"  }
\]

The S18 protein pattern is first translated into regular expression, and is passed to the
`POSIX Reverse Matching engine' in step A.
\[
\mbox{
``[IVRLP][DYN][YLF].\{2,3\}...[RHG][LIVMASR]"  }
\]
Here multigrams that match the regular expression are retrieved
from index database in step B and C. For example, `IDY', `YYX', ... etc.
In step D, the posting list will generate a list of candidate
protein database records that will potentially match the query S18.
Finally, in step E and F, the regular expression rule matching is
performed, which will produce the query result.

\begin{figure}[t]
    \centering
    \includegraphics[width=1.0\textwidth]{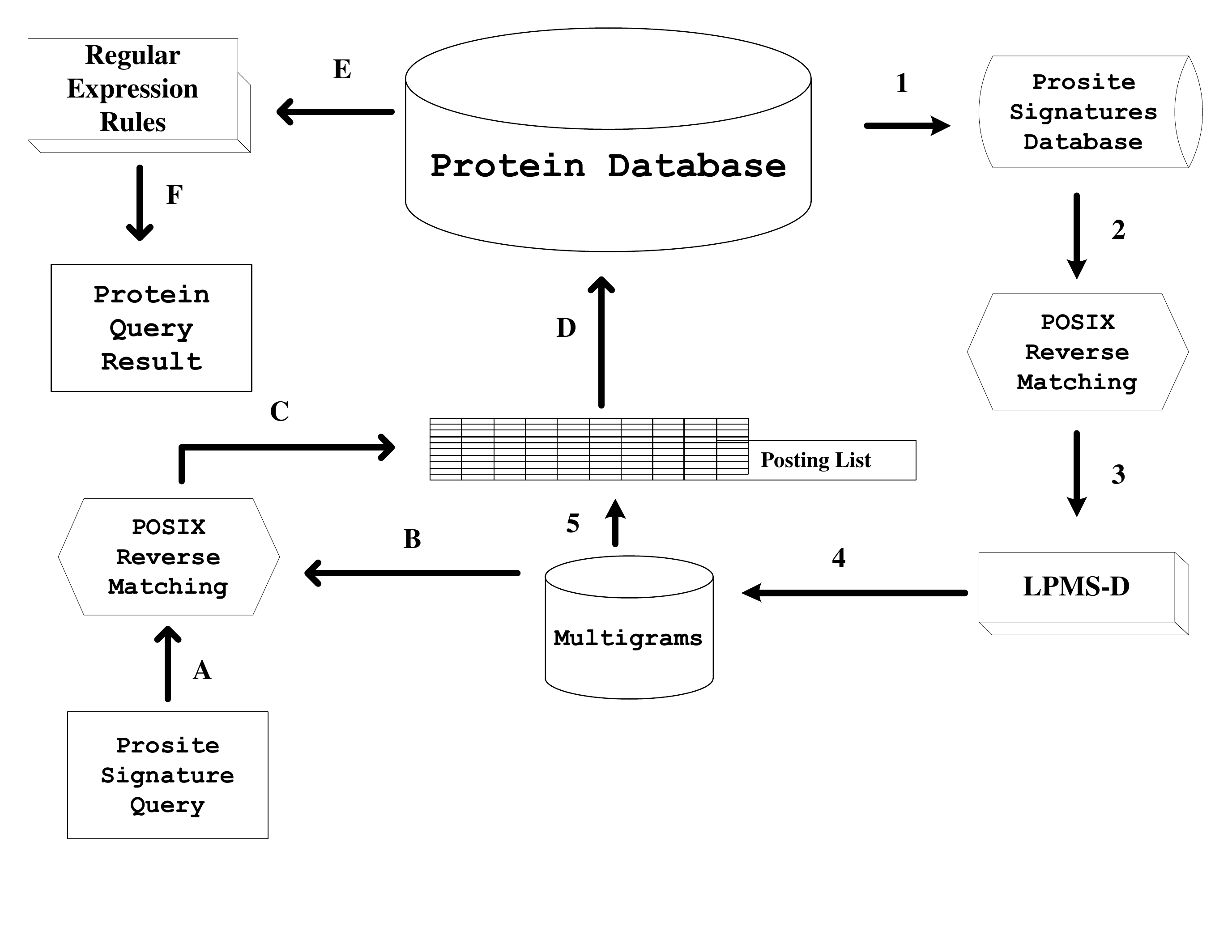}
    \caption{Data flow diagram for the Full POSIX Prosite-Protein application.
The data flow step 1 to step 5 on the right show the LPMS-D Multigram index construction process.
The data flow step A to step E on the left show the LPMS-D Multigram index query process.}
    \label{fig:full_posix_app}
\end{figure}


One hundred Prosite queries were tested against the protein database with and
without the LPMS-D index.  The recall of this test was equal to $1$, which
demonstrates that the LPMS-D index is $100\%$ accurate.

Furthermore, we repeat the experiment by restricting the multigram set allowable
in the index construction to be of size at least three. As expected the recall
was lower and equal to $0.56$. This suggests that by restricting the size of the universe
of allowed multigrams we can balance the trade-off between accuracy and efficiency.

\section{Summary and Conclusion}
While modern database management systems support forms of regular 
expression querying, they do not provide any indexing support for such
queries. Thus, a regular expression query requires 
a full database scan to find the matching records. This is a severe 
limitation as the database size will continue to increase and 
applications for such queries (e.g., bioinformatics) proliferate.

In this paper, we have proposed a robust, scalable and efficient approach
to design an index for such queries. The heart of our approach is
to model the multigram selection problem as an integer program (IP) and 
show that the approximate solutions of the IP have many of the properties we desire: accuracy
robustness and efficiency.  Extensive set of experiments on both synthetic
and real datasets demonstrate our claimed contributions.

For future work we will replace the current linear programming solver by using
a primal-dual approach which will make the approach handle very large query workloads.
Furthermore we will test our approach in other application domains including intrusion
detection systems.

\bibliographystyle{abbrv}
\bibliography{REGEX}  

\section{Appendix 1: Data Sets}
\subsection{Experiment 1}
We generated five data sets using the English alphabet by varying the support distribution
of the multigrams. The different support distribution was created as a function of 
a Normal distribution by varying the standard deviation between 100 and 500. The support
distribution of the multigrams is shown in Figure~\ref{fig:expt1_profile}. 

Based on the support distribution of the multigrams, synthetic databases were reverse engineered
which contained multigrams of prescribed support. Finally, using a query pattern consisting of 
three query keys and two gap constraints, the query workload was generated. The size of the five
databases ranged from 380 - 420K and the size of the query workload ranged between 227 and 248.

\begin{figure}[h]
    \centering
    \includegraphics[width=1.0\textwidth]{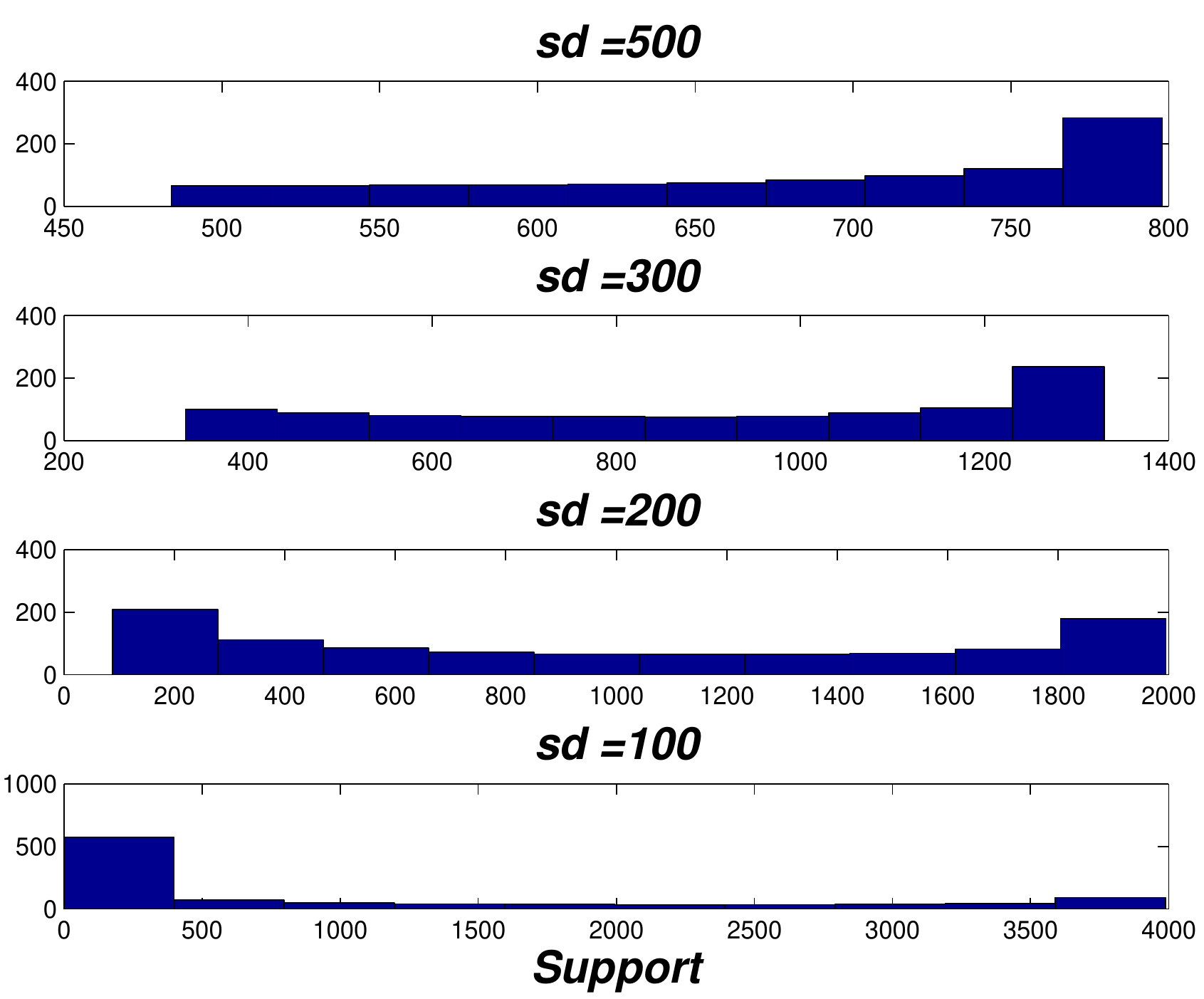}
    \caption{Experiment 1: The support distribution of multigrams was generated and both
database and query workload were reverse engineered to approximately match the multigram  distribution}
    \label{fig:expt1_profile}
\end{figure}

\subsection{Experiment 2}
Five data sets were synthetically generated using a random number generator on the English alphabet. 
Then a database sample of $10\%$ was taken from each of the data set to generate
a query work load. No distribution (of multigrams) was prescribed as in Experiment 1.
The query workload was fixed (to 1000)  and then database size was varied and then
the database size was fixed and the query workload was between between $90K$ and $500K$.
However, the size of the data sets are varied
by increasing changing the data set size, while the query work load is
kept to $1,000$. 
\subsection{Experiment 3}
To test for optimality we sampled a data set generated from Experiment 2.
The data set size was reduced to $20K$ and the query workload to $100$. This
was necessary as integer programming solutions require exponential time in
most instances.

\subsection{Experiment 4}
Data set for the robustness experiment was again 
synthetically generated on the English alphabet. Specifically, data sets 
Rob01, Rob02, Rob03 and Rob04  are composed of
the alphabet `A-D', `A-H', `A-L' and `A-P' respectively.
Each data set contains $5,000$ records with their total database
size and mean record sets are volumetrically similar. 

Next, in experiment 4, three database samples, $10\%$, $30\%$ and $50\%$
are taken from the four Rob data sets which form the basis of twelve LPMS-D
and twelve LPMS-R multigram indexes.   
\end{document}